\documentclass[aps,pra,twocolumn,showpacs]{revtex4}
\usepackage{amsmath}
\usepackage{amsthm}
\usepackage{graphicx}
\usepackage{dcolumn}
\usepackage{bm}

\newtheorem{theorem}{Theorem}

\begin{document}

\title{A type of localized quadripartite entanglement}
\author{Shao-xiong Wu, and Chang-shui Yu\footnote{quaninformation@sina.com;\quad ycs@dlut.edu.cn}}
\affiliation{School of Physics and Optoelectronic Technology, Dalian University of
Technology, Dalian 116024, China }
\date{\today}

\begin{abstract}
In this paper, we show that the average three-tangle of the reduced
tripartite density matrix for some quadripartite pure states can be
increased by some potential measurements on the fourth subsystem,
which means this type of quadripartite entanglement can be
localized. In particular, we prove that the maximal increment with
all potential measurements taken into account is a quadripartite
entanglement monotone, so it quantifies this localized quadripartite
entanglement. By analyzing quadripartite pure states based on the
previous classification, we find that this quadripartite
entanglement monotone is not only present in the standard GHZ state
and absent in the standard W states. In addition, based on the
proposed entanglement measure, we construct a new monogamy relation.
\end{abstract}

\maketitle

\section{Introduction}

Quantum entanglement is an important physical resource in quantum
information processing \cite{rmp}. Quantification of entanglement is
one of the most important tasks in the quantum information theory.
However, although it has been more than ten years since the first
several remarkable works on the bipartite entanglement
\cite{peres,Horodecki,wott97,wott98}, we can safely say that the
good understanding of quantum entanglement is still restricted to
bipartite low-dimensional systems(see the Ref. \cite{rmp} and their
references). The quantification of the multipartite entanglement
remains an open question \cite{concurrence}.

One of the important difficulties in the quantification of
multipartite entanglement is that multipartite quantum entanglement
can be classified into different inequivalent classes
\cite{santifenlianglei,Fenlei}, which implies that 1) in general a
single quantity is not enough to characterize the multipartite
quantum entanglement; 2) there exist different approaches to
classifying the multipartite entanglement. However, that isn't to
say that a single quantify is useless, because we can always use a
single quantity to quantify the entanglement of some given types.
The most obvious example is that, although tripartite quantum pure
states of qubits can be entangled in two different ways, 3-tangle is
a quite successful measure of GHZ type entanglement \cite{ckw}. In
addition, a good entanglement measure can not only distinguish the
given entangled state from others (in usual, the separable states),
but also should be an entanglement monotone which is not increased
under stochastic local operations and classical
communications(SLOCC) \cite{monotone}. Therefore, although many
authors have classified the multipartite quantum states into various
inequivalent classes, they can not provide a good and direct
entanglement measure for each class of quantum states
\cite{ustccao,Japan,Japan1,fanxianxing,duoxiangshi,lidafa,symmetric,String,Spain,ntangle}.
In particular, in order to measure the GHZ type entanglement of
multipartite quantum states, some authors defined an n-tangle for
multipartite quantum pure states to generalize the 3-tangle
\cite{ntangle}. However, it is regretful that for quadripartite
quantum states of qubits, the so-called n-tangle has value one for
the product of two Bell states, even though the n-tangle is an
entanglement monotone.

In usual, the classification of multipartite entanglement is based
on the SLOCC. In this sense, tripartite pure states of qubits can be
classified into GHZ state and W state \cite{santifenlianglei} and
quadripartite pure states of qubits can be classified into nine
families \cite{Fenlei}. That means that a general quantum
quadripartite pure state can be converted into the standard quantum
state of the particular family by SLOCC. However, a particular
characteristic of the tripartite GHZ states and the tripartite W
states is that some measurements on one qubit of GHZ state could
increase the average entanglement of the residual two qubits. But
one can not find any measurements on one qubit of W state such that
the average entanglement of the residual two qubits is increased. In
other words, the tripartite entanglement in GHZ state can be
localized to the bipartite entanglement if the SLOCC is available,
but the tripartite entanglement in W state has no such property. Can
we consider the quadripartite entanglement in the same manner?

In this paper, we consider the above question by studying the
quadripartite quantum pure states of qubits. Based on whether the
quadripartite entanglement can be localized under SLOCC, i.e. whether some measurements can increase the average GHZ type entanglement of the residual three quits, we propose
a quadripartite entanglement monotone for $\left( 2\otimes 2\otimes
2\otimes n\right)$-dimensional quantum pure states. Comparing to the
conventional classification of quadripartite quantum pure states of
qubits, we find that not only the standard quadripartite GHZ state can be
localized and the W state is not the unique quadripartite entangled
state which can not be localized.  In addition, based on the quadripartite entanglement
measure, we find an interesting monogamy relation. This paper is organized as follows. In Sec. II, we briefly
introduce the fundamental definitions of various entanglement that
will be used in the paper. In Sec. III, we give our entanglement
measure and prove that it is an entanglement monotone. In Sec. IV, compared with the previous classification, we
analyze  our entanglement measure on each family. In Sec V, we give a new monogamy relationship for
quadripartite pure states. Finally, the conclusion is drawn.

\section{Definitions}
The  three-tangle $\mu _{3}(\left\vert \psi \right\rangle )$ for a
tripartite pure state of qubits $|\psi \rangle
=\sum_{i,j,k=0}^1a_{ijk}\left\vert ijk\right\rangle $ can be given
by \cite{ckw}
\begin{equation}
\mu _{3}=4|d_{1}-2d_{2}+4d_{3}|,
\end{equation}%
with
\begin{eqnarray}
d_{1}&=&a_{000}^{2}a_{111}^{2}+a_{001}^{2}a_{110}^{2}+a_{010}^{2}a_{101}^{2}+a_{100}^{2}a_{011}^{2},\cr
d_{2}&=&a_{000}a_{111}a_{011}a_{100}+a_{000}a_{111}a_{101}a_{010}
\cr &+&a_{000}a_{111}a_{110}a_{001}+a_{011}a_{100}a_{101}a_{010} \cr
&+&a_{011}a_{100}a_{110}a_{001}+a_{101}a_{010}a_{110}a_{001}, \cr
d_{3} &=&a_{000}a_{110}a_{101}a_{011}+a_{111}a_{001}a_{010}a_{100}.
\end{eqnarray}%
It is easy to find that $\tau _{3}=\sqrt{\mu _{3}}$ is also a good
entanglement measure \cite{santifenlianglei}. Thus for a tripartite mixed state $\rho $,
the 3-tangle can be defined, based on the convex-roof-construction,
as
\begin{equation}
\tau _{3}\left( \rho \right) =\min_{\{p_{i},\pi _{i}\}}\sum_{i}p_{i}\tau
_{3}(\pi _{i}),
\end{equation}%
with $\{p_{i},\pi _{i}\}$ is a possible realization of $\rho $. In
particular, if for a quadripartite $\left( 2\otimes 2\otimes
2\otimes n\right) $-dimensional quantum pure state $\left\vert \psi
\right\rangle _{ABCS}$, $\rho =Tr_{S}\left\vert \psi \right\rangle
_{ABCS}\left\langle \psi \right\vert $, we can define the 3-tangle
of assistance as \cite{renxijun}
\begin{equation}
\tau _{a}\left( \rho \right) =\max_{\{p_{i},\pi
_{i}\}}\sum_{i}p_{i}\tau _{3}(\pi _{i}).      \label{fuzhu3tangle}
\end{equation}%
Since for a tripartite quantum pure state of qubits $\left\vert \phi
\right\rangle ,$
\begin{equation}
\tau _{3}(A\otimes B\otimes C\left\vert \phi \right\rangle
)=\left\vert \det \left( A\right) \right\vert \left\vert \det \left(
B\right) \right\vert \left\vert \det \left( C\right) \right\vert
\tau _{3}(\left\vert \phi \right\rangle )    \label{zhibubian}
\end{equation}%
holds for any local operations $A$, $B$ and $C$, a similar proof as
Ref. \cite{G.Gour} can show that the three-tangle of assistance defined
in Eq. (\ref{fuzhu3tangle}) is a quadripartite entanglement
monotone.

\section{Quadripartite entanglement monotone}

Let $\left\vert \psi \right\rangle _{ABCS}$ be a quadripartite ($2\otimes
2\otimes 2\otimes n$)-dimensional quantum pure state, $\rho
_{ABC}=Tr_{S}\left\vert \psi \right\rangle _{ABCS}\left\langle \psi
\right\vert $, based on the definitions we can easily find that $\tau
_{a}(\rho _{ABC})\geq \tau _{3}(\rho _{ABC})$. Define

\begin{equation}
\tau _{4}\left( \left\vert \psi \right\rangle _{ABCS}\right) =\tau
_{4}\left[ \rho _{ABC}\right] =\sqrt{\tau _{a}^{2}(\rho _{ABC})-\tau
_{3}^{2}(\rho _{ABC})}  \label{diyigedanpei}
\end{equation}%
we can find that $\tau _{4}\left( \left\vert \psi \right\rangle
_{ABCS}\right) $ has the following properties.

\bigskip i) $\tau _{4}\left( \left\vert \psi \right\rangle _{ABCS}\right) $
is invariant under determinant-one operations on subsystems A, B and C.

ii) $\tau _{4}\left( \left\vert \psi \right\rangle _{ABCS}\right) $ is a
concave function on $\rho _{ABC}$.

One can find that the property i) is very obvious based on Eq.
(\ref{zhibubian}). Now we will prove the property ii). Let $\lambda
\in \lbrack 0,1]$ and $\rho _{1},\rho _{2}$ be the tripartite
quantum states of qubits \cite{Yu}, then
\begin{eqnarray}
&&\lambda \tau _{4}\left[ \rho _{1}\right] +(1-\lambda )\tau _{4}\left[ \rho
_{2}\right]  \notag \\
&=&\lambda \sqrt{\left[ \tau _{a}(\rho _{1})-\tau _{3}(\rho _{1})\right] %
\left[ \tau _{a}(\rho _{1})+\tau _{3}(\rho _{1})\right] }  \notag \\
&&+(1-\lambda )\sqrt{\left[ \tau _{a}(\rho _{2})-\tau _{3}(\rho _{2})\right] %
\left[ \tau _{a}(\rho _{2})+\tau _{3}(\rho _{2})\right] }  \notag \\
&\leq &\sqrt{\lambda \left[ \tau _{a}(\rho _{1})-\tau _{3}(\rho _{1})\right]
+(1-\lambda )\left[ \tau _{a}(\rho _{2})-\tau _{3}(\rho _{2})\right] }
\notag \\
&&\times \sqrt{\lambda \left[ \tau _{a}(\rho _{1})+\tau _{3}(\rho _{1})%
\right] +(1-\lambda )\left[ \tau _{a}(\rho _{2})+\tau _{3}(\rho _{2})\right]
}  \notag \\
&=&\sqrt{\left[ \lambda \tau _{a}(\rho _{1})+(1-\lambda )\tau _{a}(\rho _{2})%
\right] ^{2}-\left[ \lambda \tau _{3}(\rho _{1})+(1-\lambda )\tau _{3}(\rho
_{2})\right] ^{2}}  \notag \\
&\leq &\sqrt{\tau _{a}^{2}\left[ \lambda \rho _{1}+(1-\lambda )\rho _{2}%
\right] -\tau _{3}^{2}\left[ \lambda \rho _{1}+(1-\lambda )\rho _{2}\right] }%
\label{convex}
\end{eqnarray}%
where the first inequality holds based on Cauchy-Schwarz
inequality:$\sum_i{x_iy_i} \leq
(\sum_i{x_i^2})^{1/2}(\sum_j{y_j^2})^{1/2}$ and the second
inequality holds based on the definition of $\tau _{a}$ and $\tau
_{3} $.

\begin{theorem}
$\tau _{4}\left( \left\vert \psi \right\rangle _{ABCS}\right) $ is an
quadripartite entanglement monotone. It can measure one type
of localized quadripartite entanglement.
\end{theorem}

\begin{proof}
We will first prove $\tau _{4}\left( \left\vert \psi \right\rangle
_{ABCS}\right) $ does not increase under SLOCC operations.  We
suppose the pure state is shared by A, B, C and S. Let us consider
the following four-way SLOCC. Firstly, S performs a measurement
denoted by the Kraus operators $M_{i} $ on the qudit S and send the
result $i$ to others; Secondly, based on the result of S, A performs
a measurement $A_{j}^{i}$ on qubit A and send his result $j$ to
others; Thirdly, based on the previous results of A and S, Bob
performs a measurement $B_{k}^{ij}$ on qubit B and send the result
$k$ to others, and similarly, based on the results A, B and S,
Charlie performs a measurement $C_{l}^{ijk}$ on qubit C and send the
result $l$ to others;
Finally, based on all the previous results, S performs a measurement $%
F_{m}^{ijkl}$ on qudit S and send his result $m$ to others. After
the SLOCC operations, the initial state $\left\vert \psi
\right\rangle _{ABCS}$ will become a mixed state as
$\{N_{ijklm},\left\vert \psi \right\rangle _{ABCS}^{ijklm}\}$, where
\begin{equation}
\left\vert \psi \right\rangle _{ABCS}^{ijklm}=\frac{A_{j}^{i}\otimes
B_{k}^{ij}\otimes C_{l}^{ijk}\otimes F_{m}^{ijkl}M_{i}}{\sqrt{N_{ijklm}}}%
\left\vert \psi \right\rangle _{ABCS}
\end{equation}%
and $N_{ijklm}$ is the corresponding probability. Therefore, we can obtain
the average $\tau _{4}$ as%
\begin{gather}
\sum_{ijklm}N_{ijklm}\tau _{4}\left( \left\vert \psi \right\rangle
_{ABCS}^{ijklm}\right) =\sum_{ijklm}D_{ijkl}\tau _{4}\left(
F_{m}^{ijkl}M_{i}\left\vert \psi \right\rangle _{ABCS}\right)   \notag \\
=\sum_{ijklm}D_{ijkl}\tau _{4}\left[
Tr_{S}F_{m}^{ijkl}M_{i}\left\vert \psi \right\rangle
_{ABCS}\left\langle \psi \right\vert M_{i}^{\dag }\left(
F_{m}^{ijkl}\right) ^{\dag }\right] ,   \label{rhs}
\end{gather}%
with $D_{ijkl}=\left\vert \det \left( A_{j}^{i}\right) \right\vert
\left\vert \det \left( B_{k}^{ij}\right) \right\vert \left\vert \det \left(
C_{l}^{ijk}\right) \right\vert $. Let $\tilde{\rho}_{i}=Tr_{S}M_{i}\left%
\vert \psi \right\rangle _{ABCS}\left\langle \psi \right\vert M_{i}^{\dag }$%
. Since the reduced density matrix seen from A, B and C can not be
changed by the local operations $F_{m}^{ijkl}$ alone, we have $\tilde{\rho}%
_{i}=\sum_{m}\sigma _{m}$, with $\sigma
_{m}=Tr_{S}F_{m}^{ijkl}M_{i}\left\vert \psi \right\rangle
_{ABCS}\left\langle \psi \right\vert M_{i}^{\dag }\left( F_{m}^{ijkl}\right)
^{\dag }$. Based on the property ii), one can easily find that the right
hand side (r.h.s.) of Eq. (\ref{rhs}) can lead to the following inequality:%
\begin{equation}
r.h.s\leq \sum_{ijklm}D_{ijkl}\tau _{4}\left[ Tr_{S}M_{i}\left\vert \psi
\right\rangle _{ABCS}\left\langle \psi \right\vert M_{i}^{\dag }\right] .
\end{equation}%
According to the geometric-arithmetic inequality: $\sum\limits_{j}\left\vert
\det A_{j}^{i}\right\vert \leq \frac{1}{2}\sum\limits_{j}TrA_{j}^{i}A_{j}^{i%
\dag }=1$ and the similar inequalities for $B_{k}^{ij}$ and $C_{l}^{ijk}$,
we can get%
\begin{eqnarray}
r.h.s &\leq &\sum_{i}\tau _{4}\left[ Tr_{S}M_{i}\left\vert \psi
\right\rangle _{ABCS}\left\langle \psi \right\vert M_{i}^{\dag }\right]
\notag \\
&\leq &\tau _{4}\left[ \rho _{ABC}\right]   \label{zhengmingdandiao}
\end{eqnarray}%
where the second inequality holds similar to the Eq. (\ref{convex}).
Eq. (\ref{zhengmingdandiao}) shows that $\tau _{4}$ is a
quadripartite entanglement monotone (similar as Ref. \cite{G.Gour}.)

In addition, we claim that our entanglement monotone is a quadripartite
entanglement measure, we will prove that $\tau _{4}$ vanishes for any
separable state. A quadripartite separable quantum pure state can be given
as $\left\vert \varphi _{1}\right\rangle _{ABCD}=\left\vert \chi
_{1}\right\rangle _{A}\otimes \left\vert \gamma _{1}\right\rangle _{BCD}$, $%
\left\vert \varphi _{2}\right\rangle _{ABCD}=\left\vert \chi
_{2}\right\rangle _{AB}\otimes \left\vert \gamma _{2}\right\rangle _{CD}$
and $\left\vert \varphi _{3}\right\rangle _{ABCD}=\left\vert \chi
_{3}\right\rangle _{ABC}\otimes \left\vert \gamma _{3}\right\rangle _{D}$
where the state $\left\vert \cdot \right\rangle $ with multipartite
subscripts out of the ket denotes a general multipartite quantum pure state,
so a completely separable state is also included. For $\left\vert \varphi
_{1}\right\rangle _{ABCD}$, if we trace out the subsystem $D$, the final
reduced density can be given by $\sigma _{1}^{ABC}=\left\vert \chi
_{1}\right\rangle _{A}\left\langle \chi _{1}\right\vert \otimes \varrho
_{BC} $. Hence, we can find $\tau _{a}(\sigma _{1}^{ABC})=\tau _{3}(\sigma
_{1}^{ABC})=0$, which shows $\tau _{4}\left( \left\vert \varphi
_{1}\right\rangle _{ABCD}\right) =0$. Similarly, one can find that $\tau
_{4}\left( \left\vert \varphi _{2}\right\rangle _{ABCD}\right) =0$. For $%
\left\vert \varphi _{3}\right\rangle _{ABCD}$, when we trace out the
subsystem D, we can obtain the reduced density matrix as $\sigma
_{1}^{ABC}=\left\vert \chi _{3}\right\rangle _{ABC}\left\langle \chi
_{3}\right\vert $. So it is easy to find that $\tau _{a}(\sigma
_{1}^{ABC})=\tau _{3}(\sigma _{1}^{ABC})$, which also means that
$\tau _{4}\left( \left\vert \varphi _{3}\right\rangle _{ABCD}\right)
=0$.

One can directly find that our Eq. (\ref{diyigedanpei}) actually
implies a relationship of monogamy, which shows that $\tau _{4}$ can
lead to the increasing of the average 3-tangle (the maximum is the
3-tangle of assistance). The definition of $\tau _{4}$ obvioulsy
shows that the entanglement quantified by $\tau _{4} $ can be
localized.
\end{proof}

\section{Analyze the quantum states of different families using the entanglement measure}

We have shown that $\tau _{4}$ can quantify one type of localized quadripartite
entanglement which is the generalization of some property of the
tripartite GHZ state. After all, there are not only the GHZ state
for quadripartite states. So in this section, we would like to
discuss which states of four qubits hold this type of so-called localized
entanglement. As we know, quadripartite quantum pure state of qubits
has nine families up to the permutation of the qubits in the sense
of SLOCC. The standard state of each family can be given as
follows \cite{Fenlei}.
\begin{eqnarray}
G_{abcd} &=&\frac{a+d}{2}(|0000\rangle +|1111\rangle )+\frac{a-d}{2}
(|0011\rangle +|1100\rangle )   \cr
 &+&\frac{b+c}{2}(|0101\rangle
+|1010\rangle )+\frac{b-c}{2}(|0110\rangle +|1001\rangle ), \cr
L_{abc_{2}} &=&\frac{a+b}{2}(|0000\rangle +|1111\rangle
)+\frac{a-b}{2} (|0011\rangle +|1100\rangle ) \cr &+&c(|0101\rangle
+|1010\rangle )+|0110\rangle , \cr
 L_{a_{2}b_{2}} &=&a(|0000\rangle +|1111\rangle
)+b(|0101\rangle +|1010\rangle )   \cr
 &+&|0110\rangle +|0011\rangle
, \cr
 L_{ab_{3}} &=&a(|0000\rangle
+|1111\rangle )+\frac{a+b}{2}(|0101\rangle +|1010\rangle )  \cr
&+&\frac{i}{\sqrt{2}}(|0001\rangle +|0010\rangle +|0111\rangle
+|1011\rangle ) \cr &+&\frac{a-b}{2}(|0110\rangle +|1001\rangle ),
\cr
 L_{a_{4}}
&=&a(|0000\rangle +|0101\rangle +|1010\rangle +|1111\rangle ) \notag
\cr
 &+&(i|0001\rangle +|0110\rangle
-|1011\rangle ), \cr L_{a_{2}0_{3\oplus 1}} &=&a(|0000\rangle
+|1111\rangle )+(|0011\rangle +|0101\rangle +|0110\rangle ), \cr
 L_{0_{5\oplus
\bar{3}}} &=&|0000\rangle +|0101\rangle +|1000\rangle +|1110\rangle
, \cr
 L_{0_{7\oplus \bar{1}}} &=&|0000\rangle
+|1011\rangle +|1101\rangle +|1110\rangle , \cr
 L_{0_{3\oplus \bar{1}}0_{3\oplus \bar{1}}}
&=&|0000\rangle +|0111\rangle .
\end{eqnarray}%

It is shown that any a quadripartite quantum state of qubits can be
converted into one of the above standard states by SLOCC. For the
convenience, we will study our proposed localized entanglement by tracing the
first qubit. In particular, we will show what kind of standard
quantum states have no this type entanglement.

For $G_{abcd}$, one can find that $\tau_{4}(G_{{abcd}})=0$ for $%
|a|=|b|=|c|=|d|$ or any three of the four parameters vanish. In these cases,
$G_{abcd}$ will become separable states. The state including two EPR pairs
belongs to this family with three vanishing parameters. For $L_{abc_2}$, it
is easy to see that $\tau_{4}(L_{{abc_{2}}})=0$ for $|a|=|c|$ or $|b|=|c|$
which shows that $L_{abc_2}$ is not a separable state but a special
quadripartite quantum state. For $L_{a_2b_2}$, if $|a|=|b|$, then $%
\tau_{4}(L_{{a_{2}b_{2}}})=0$. It is interesting that if $a=b=0$, $%
L_{a_2b_2} $ denotes a separable state, otherwise, it is a quadripartite
entangled state. For $L_{ab_3}$, one can find that $\tau_4(L_{a b_3})$
vanishes for $|a|=|b|$, even in this case, $L_{a b_3}$ is not separable. In
particular, can find that the quadripartite W state $\left\vert
W_4\right\rangle=\frac{1}{2}\left(\left\vert 0001\right\rangle+\left\vert
0010\right\rangle +\left\vert 0100\right\rangle +\left\vert 1000\right
\rangle\right)$ belongs to this family with $a=b=0$. So $\tau_4(\left\vert
W_4\right\rangle)=0$. If $a=0$, one can show that $\tau_{4}(L_{a_4})$ and $%
\tau_{4}(L_{a_20_{3\oplus1}})$ vanish. But in this case, $L_{a_4}$ is not
separable, but $L_{a_20_{3\oplus1}}$ is separable. In addition, we can find
that $L_{0_{5\oplus\bar{3}}}$ and $L_{0_{3\oplus\bar{1}}0_{3\oplus\bar{1}}}$
have no this type of localized entanglement, even though $L_{0_{5\oplus\bar{3}}}$ is not
separable. On the contrary, $L_{0_{7\oplus\bar{1}}}$ always has nonzero the localized quadripartite
entanglement.

Thus we have divided the nine families into two parts based on our
requirements. However, one could ask a natural question whether our measure
depends on the permutation of qubits. It is unfortunate that our definition
strongly depends on which qubit is traced out. A simple example is the state
$L_{0_{7\oplus \bar{1}}}$. One can find that $\tau _{4}$ for this state will
vanish if we trace out the second, the third or the fourth qubit, even
though we have shown it is not the case if we trace out the first qubit.
Therefore, from the point of entanglement measure of view, we have to use a
vector to completely characterize the type of localized entanglement of a
quadripartite pure state. Here one can define an entanglement vector as $%
\tau _{4}=[\tau _{(A)BCD},\tau _{A(B)CD},\tau _{AB(C)D},\tau _{ABC(D)}]$,
with bracket in subscripts denoting trace over the corresponding subsystem.
Based on the entanglement vector, one can further consider what type the
entanglement of a quadripartite pure state is.

\section{a new monogamy}

Besides the definition which has presented a monogamy relation, we can find another new monogamy. For a pure state $\left\vert \phi \right\rangle _{ABCD}$, consider an
optimal decomposition $\{p_{i},\left\vert \psi _{BCD}\right\rangle _{i}\}$
of $\rho _{BCD}=Tr_{A}\left\vert \phi \right\rangle \left\langle \phi
\right\vert $ such that $\tau _{a}(\rho _{BCD})=\sum\limits_{i}p_{i}\tau
_{3}(\left\vert \psi _{BCD}\right\rangle _{i})$, we will have%
\begin{equation}
\tau _{(A)BCD}=\sqrt{\left[ \sum\limits_{i}p_{i}\tau _{3}(\left\vert
\psi _{BCD}\right\rangle _{i})\right] ^{2}-\tau _{3}^{2}(\rho
_{BCD})}.\label{tua-abcd}
\end {equation}%
Since for each tripartite pure state $\left\vert \psi
_{BCD}\right\rangle $, we have\cite{Yu}
\begin{equation}
\tau _{3}(\left\vert \psi _{BCD}\right\rangle =\sqrt{C_{a}^{2}(\rho
_{CD})-C^{2}(\rho _{CD})},\label{tau3}
\end{equation}%
where $\rho _{CD}=Tr_{B}\left\vert \psi _{BCD}\right\rangle \left\langle
\psi _{BCD}\right\vert $ and $C(\rho _{CD})$ is the concurrence and $%
C_{a}(\rho _{CD})$ is the concurrence of assistance. Substitute
Eq. (\ref{tau3}) into Eq. (\ref{tua-abcd}), we will obtain%
\begin{eqnarray}
\tau _{(A)BCD}^{2}+\tau _{3}^{2}(\rho _{BCD}) &=&\left[ \sum\limits_{i}p_{i}%
\sqrt{C_{a}^{2}(\rho _{CD}^{i})-C^{2}(\rho _{CD}^{i})}\right] ^{2}  \notag \\
&\leq &C_{a}^{2}(\rho _{CD})-C^{2}(\rho _{CD})  \notag \\
&=&\tau _{3}^{2}(\left\vert \phi \right\rangle _{[AB]CD})
\label{monogamy}
\end{eqnarray}%
where $\tau _{3}(\left\vert \phi \right\rangle _{[AB]CD})$ denotes
the 3-tangle by considering subsystems A and B as one party. In
addtition, the inequality holds based on Cauchy-Schwarz inequality.
This inequality is not trivial, because an intuitional observation
can show that the inequality can be saturated if $\left\vert \phi
\right\rangle _{[AB]CD}$ is the
standard quadripartite GHZ state $\left\vert GHZ\right\rangle _{4}=\frac{1}{%
\sqrt{2}}\left( \left\vert 0000\right\rangle +\left\vert
1111\right\rangle \right) $. Thus Eq. (\ref{monogamy}) provides us
with a new inequality of monogamy. Analogously,
 we can also obtain other type inequalities of monogamy by considering the trace over
 different subsystems. These inequalities have the similar nature, so we omit them.

\section{conclusions and discussion}

In this paper, we have found a quadripartite entanglement monotone which quantifies one
type of localized quadripartite entanglement. That is, for one kind of quadripartite quantum pure state, the proposed localized quadripartite entanglement can converted to 3 subsystems by some measurements on the fourth subsystem. Along this line, we find that not only the standard GHZ state of four quits have nonzero localized entanglement and not only the standard W state has no localized entanglement, which is quite  different from the case of tripartite pure states. However, it is actually consistent with that quadripartite quantum entanglement can be classified into more than 2 classes. In addition, a new monogamy besides the definition itself has also been proposed.
At last, we would like to emphasize it is a quite interesting question what kind of other monogamy relations can be constructed based on the current localized entanglement measure or some others similar.

\section{Acknowledgement}

This work was supported by the National Natural Science Foundation of China,
under Grant No. 10805007 and No. 10875020, and the Doctoral Startup
Foundation of Liaoning Province.

\end{document}